\documentclass[12pt,a4paper]{article}

\usepackage{latexsym}
\usepackage{amsmath}
\usepackage{amssymb}
\usepackage{arydshln}

\usepackage{cite}
\usepackage{stmaryrd}
\usepackage{enumerate}

\usepackage{hyperref}

\usepackage{color}
\usepackage{lineno}
\usepackage{graphicx}
\usepackage{ae}
\usepackage{amsmath}
\usepackage{amssymb}
\usepackage{latexsym}
\usepackage{url}
\usepackage{epsfig}
\usepackage{mathrsfs}
\usepackage{amsfonts}
\usepackage{amsthm}
\def\mod{{\rm mod}}

\newtheorem{theorem}{Theorem}[section]

\newtheorem{prop}[theorem]{Proposition}
\newtheorem{lemma}[theorem]{Lemma}

\newtheorem{cor}[theorem]{Corollary}

\newtheorem{example}{Example}

\theoremstyle{definition}
\newtheorem{definition}{Definition}

\numberwithin{equation}{section} %%��ʽ���Ÿ����½�
\allowdisplaybreaks  %%��ʽ��ҳչʾ

\def\qed{\hfill$\Box$\vspace{12pt}}

\long\def\delete#1{}

\usepackage{xcolor}
\usepackage[normalem]{ulem}

\usepackage{tikz}
\usetikzlibrary{decorations.markings}
\tikzstyle{vertex}=[circle, inner sep=1.2pt, minimum size=3pt]
\tikzstyle{filledvertex}=[circle, draw, fill, inner sep=1.2pt, minimum size=3pt]

\newcommand{\vertex}{\node[vertex]}

\tikzstyle{directed}=[postaction={decorate,
	decoration={markings,mark=at position 0.5 with {\arrow{stealth}}}
}]
\begin{document}
\title{Laplacian state transfer in vertex complemented coronas}
\author{Qiang Wang$^{a,b}$,~Xiaogang Liu$^{a,c,}$\thanks{Supported by the National Natural Science Foundation of China (Nos. 11601431 and 11871398), the Natural Science Foundation of Shaanxi Province (No. 2020JM-099), and the Natural Science Foundation of Qinghai Province  (No. 2020-ZJ-920).}~$^,$\thanks{ Corresponding author. Email addresses: qiang47163606@mail.nwpu.edu.cn, xiaogliu@nwpu.edu.cn}   %, sanming@unimelb.edu.au
\\[2mm]     %  ~,~Sanming Zhou$^{d}$
{\small $^a$School of Mathematics and Statistics,}\\[-0.8ex]
{\small Northwestern Polytechnical University, Xi'an, Shaanxi 710072, P.R.~China}\\
{\small $^b$College of Mathematics and Computer Science,}\\[-0.8ex]
{\small Yan'an University, Yan'an, Shaanxi 716000, P.R. China}\\
{\small $^c$School of Mathematics and Statistics,}\\[-0.8ex]
{\small Qinghai Nationalities University, Xining, Qinghai 810007, P.R. China}
}

\date{}

\openup 0.5\jot
\maketitle

\begin{abstract}

In this paper, we investigate the existence of Laplacian perfect state transfer and Laplacian pretty good state transfer in vertex complemented coronas. We prove that there is no Laplacian perfect state transfer in vertex complemented coronas. In contrast, we give a sufficient condition for vertex complemented coronas to have Laplacian pretty good state transfer.

\smallskip

\emph{Keywords:} Laplacian perfect state transfer;  Laplacian pretty good state transfer; Vertex complemented corona.

\emph{Mathematics Subject Classification (2010):} 05C50, 81P68
\end{abstract}

\section{Introduction}

Let $G$ be a graph with adjacency matrix $A_G$.  The \emph{transition matrix} of $G$ with respect to $A_G$ is defined to be
$$
H_{A_{G}}(t) = \exp(-\mathrm{i}tA_{G})=\sum_{k\ge0}\frac{(-\mathrm{i})^{k} A^{k}_{G} t^{k}}{k!}, ~ t \in \mathbb{R},~\mathrm{i}=\sqrt{-1}.
$$
Motivated by the Schr\"{o}dinger equation (see \cite[Claim 3.17]{Hall13}), Farhi and Gutmann first introduced this notion and used it as a paradigm to design efficient quantum algorithms \cite{FarhiG98}. The $(u,v)$-entry of $H_{A_{G}}(t)$ is denoted by $H_{A_G}(t)_{u,v}$, where $u, v \in V(G)$. If $u \ne v$ and $|H_{A_G}(\tau)_{u,v}|=1$, then $G$ is said to have \emph{perfect state transfer} (PST for short)  from vertex $u$ to vertex $v$ at time $\tau$.  This concept, which is very important in quantum computing and quantum information processing, was first introduced by Bose in 2003 \cite{Bose03}. Since then, the problem of characterizing graphs having PST has attracted the attention of both physicists and mathematicians. Recently, it has become clear that graphs having PST are rare. Thus, Godsil posed to study a relaxation of PST, called pretty good state transfer \cite{Godsil12}. A graph $G$ is said to have \emph{pretty good state transfer}  (PGST for short)  from vertex $u$ to vertex $v$ if for each $\varepsilon>0$, there exists a time $t$ such that $$\mid H_{A_G}(t)_{u,v}\mid \geq 1-\varepsilon.$$
 Up until now, many graphs have been proved to admit PST or PGST. For details, we refer the reader to \cite{AckBCMT16, Angeles10, Basic11, CaoCL20, CaoF21, CaoWF20, CC, chris2, Coutinho15, Coh12, Fan, LiLZZ2021, SZ, TanFC19, HZ,  Zhou14} and three surveys \cite{Godsil11, CGodsil, Godsil12}.

Let $D_G$ be the diagonal matrix with diagonal entries the degrees of vertices of $G$. Let $L_G=D_G-A_G$ denote the Laplacian matrix of $G$. If we replace $A_G$ in the definition of PST (respectively, PGST) by $L_G$, then we obtain the definition of \emph{Laplacian perfect state transfer} (respectively, \emph{Laplacian pretty good state transfer}), abbreviated to LPST (respectively, LPGST). It is known that LPST  (respectively, LPGST) is as important as PST  (respectively, PGST). However, only a few results on LPST or LPGST have been given, which are listed as follows:
\begin{itemize}
  \item  The complete graph $K_n$ with a missing link, where $n$ is multiple of $4$, has LPST \cite{SBose}.
     \item  A complete characterization of the class of threshold graphs allowing for LPST was given in \cite{KirklandS11}.
  \item  Every tree with at least three vertices has no LPST \cite{Cohl}.
  \item If a graph on $n$ vertices has LPST at time $\tau$ and $n\tau \in  2\pi \mathbb{Z}$, then its complement also has LPST at time $\tau$ \cite{AlvirDLM16}.
  \item The double cone over a graph on $n$ vertices has LPST if and only if $n \equiv 2 ~\text{(\mod ~4)}$ \cite{AlvirDLM16}.
  \item The corona of two graphs has no LPST, but it has LPGST under some mild conditions \cite{Ack}.
   \item The LPGST occurs between the extremal vertices of the path with $n$ vertices if and only if $n$ is a power of $2$. Moreover, in these cases, LPGST occurs between vertices at the $j$th and $(n-j+1)$th positions for all $j = 1, \ldots, n$ \cite{Ban}.
   \item  Sufficient conditions for edge coronas to have or not have LPST, and sufficient conditions for edge coronas to have LPGST are given in \cite{LiLZ20-1}.
   \item The $Q$-graph of an $r$-regular graph, if $r+1$ is a prime number, has no LPST, but it has LPGST under some mild conditions \cite{LiLZ20-2}.

 \item If $r + 1$ is not a Laplacian eigenvalue of an $r$-regular graph $G$, then there is no LPST in the total graph of $G$, but it has LPGST under some mild conditions \cite{LiuW2021}.

 \item Sufficient conditions for edge complemented coronas to not have LPST, and sufficient conditions for edge complemented coronas to have LPGST are given in \cite{WangL2021}.

     \item  Sufficient conditions for extended neighborhood coronas to have or not have LPST are given in \cite{LiLZ2021}.

\end{itemize}

In this paper, we investigate the existence of LPST and LPGST in vertex complemented coronas, whose definition is given in Definition \ref{hian}.

\begin{definition}
\label{hian}
Let $G$ be a graph with  vertex set $V(G)=\{v_1,\ldots,v_n\}$ and let $\overrightarrow{H}=(H_1,\ldots,H_n)$
be an $n$-tuple of graphs. The \emph{vertex complemented corona}
$G\tilde{\circ}\overrightarrow{H}$ is formed by taking the
disjoint union of $G$ and $H_1,\ldots,H_n$ with each $H_i$ corresponding to the vertex $v_i$, and then joining every vertex in $H_{i}$ to every vertex in $V(G)\setminus\{v_{i}\}$ for $i=1,2,\ldots,n$.
\end{definition}

Figure \ref{VCC-Fig-1} depicts the vertex complemented corona $P_3\tilde{\circ}\overrightarrow{H}$ with $\overrightarrow{H}=(P_1,P_2,P_1)$, where $P_n$ denotes the path on $n$ vertices.

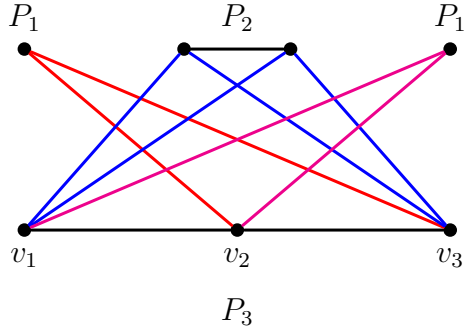
\begin{figure}
\vspace{0.5cm}
\begin{center}
\begin{tikzpicture}[x=0.70cm, y=0.6cm]
\tikzstyle{vertex}=[circle,inner sep=1.8pt, minimum size=0.1pt]

\vertex (a)[fill] at (0,0)[label=below:$v_{1}$]{};
\vertex (b)[fill] at (4,0)[label=below:$v_{2}$]{};
\vertex (c)[fill] at (8,0)[label=below:$v_{3}$]{};
\vertex (g1) at (4,-2.5)[label=above:$P_{3}$]{};
\draw[line width=.4mm,line cap=round](a)--(b);
\draw[line width=.4mm,line cap=round](b)--(c);

\vertex (d)[fill] at (0,4){};
\vertex (g2) at (0,4)[label=above:$P_{1}$]{};
\draw[line width=.4mm,line cap=round,color=red](d)--(b);
\draw[line width=.4mm,line cap=round,color=red](d)--(c);

\vertex (e)[fill] at (3,4){};
\vertex (f)[fill] at (5,4){};
\vertex (g3) at (4,4)[label=above:$P_{2}$]{};
\draw[line width=.4mm,line cap=round](e)--(f);
\draw[line width=.4mm,line cap=round,color=blue](e)--(a);
\draw[line width=.4mm,line cap=round,color=blue](e)--(c);
\draw[line width=.4mm,line cap=round,color=blue](f)--(a);
\draw[line width=.4mm,line cap=round,color=blue](f)--(c);

\vertex (g)[fill] at (8,4){};
\vertex (g4) at (8,4)[label=above:$P_{1}$]{};
\draw[line width=.4mm,line cap=round,color=magenta](g)--(a);
\draw[line width=.4mm,line cap=round,color=magenta](g)--(b);

\end{tikzpicture}
\end{center}
\vspace{-0.6cm}
\caption{An example of the vertex complemented corona}
\label{VCC-Fig-1}
\end{figure}

Our paper is organized as follows. In Section $2$, we present some lemmas that will be used later. In Section $3$, Laplacian eigenvalues and eigenprojectors of vertex complemented coronas are obtained. In Section $4$, we prove that there is no LPST in the vertex complemented coronas. In Section $5$, we give sufficient conditions for the vertex complemented coronas to have LPGST.

%%%%%%%%%%%%%%%%%%%%%%

\section{Preliminaries}\label{Sec:pre}

In this section, we list some notations and basic results which will play important roles throughout the paper.

\begin{lemma}
\emph{(see \cite{Zhang05})}\label{tian6}
Let $M_{1}$, $M_{2}$, $M_{3}$ and $M_{4}$ be respectively $p\times p$, $p\times q$, $q\times p$ and $q\times q$ matrices with $M_{1}$ and $M_{4}$ invertible. Then
     \begin{align*}
  \det\left(
        \begin{array}{cc}
          M_{1} & M_{2} \\
          M_{3} & M_{4} \\
        \end{array}
      \right) &=\det(M_{4})\cdot\det(M_{1}-M_{2}M_{4}^{-1}M_{3}) \\
              &=\det(M_{1})\cdot\det(M_{4}-M_{3}M_{1}^{-1}M_{2}),
      \end{align*}
where $M_{1}-M_{2}M_{4}^{-1}M_{3}$ and $M_{4}-M_{3}M_{1}^{-1}M_{2}$ are called the Schur complements of $M_{4}$ and $M_{1}$, respectively.
\end{lemma}

The \emph{$M$-coronal} $\Gamma_{M}(x)$ of an $n\times n$ matrix $M$ \cite{kn:Cui12, kn:McLeman11} is defined to be the sum of the entries of the matrix $(xI_{n}-M)^{-1}$, that is,
\begin{equation}\label{tian1}
  \Gamma_{M}(x)=\mathbf{j}_{n}^{\top}(xI_{n}-M)^{-1}\mathbf{j}_{n},
\end{equation}
where $\mathbf{j}_{n}$ denotes the column vector of size $n$ with all entries equal to one, and $\mathbf{j}_{n}^{\top}$ denotes the transpose of $\mathbf{j}_{n}$.

\begin{lemma}\label{tian12}
\emph{(see \cite[Proposition 2]{kn:Cui12}) }
If $M$ is an $n\times n$ matrix with each row sum equal to a constant $t$, then
\begin{equation*}
  \Gamma_{M}(x)=\frac{n}{x-t}.
\end{equation*}
In particular, since for any graph $G$ with $n$ vertices, each row sum of $L_G$ is equal to $0$, we have
\begin{equation*}
   \Gamma_{L_G}(x)=\frac{n}{x}.
\end{equation*}
\end{lemma}

\begin{lemma}
\emph{(see \cite[Corollary 2.3]{LiuZ19})}\label{liuzhang1}
Let $\alpha$ be a real number, $A$ an $n\times n$ real matrix, $I_n$ the identity matrix of size $n$, and $J_{n}$ the $n\times n$ matrix with all entries equal to one. Then
\begin{equation*}
  \det(xI_{n}-A-\alpha J_{n})=(1-\alpha \Gamma_{A}(x))\det(xI_{n}-A).
\end{equation*}
\end{lemma}

Let $G$ be a graph with its Laplacian matrix $L_G$. The eigenvalues of $L_G$ are called the \emph{Laplacian eigenvalues} of $G$. We use $\mathrm{Spec}_L(G)$ to denote the set of all distinct Laplacian eigenvalues of $G$. Suppose that $0= \lambda_0< \lambda_1<\cdots<\lambda_d$ are all distinct eigenvalues of $L_G$ and $\left\{\mathbf{x}_{1}^{(r)}, \ldots, \mathbf{x}_{l_r}^{(r)}\right\}$ is an orthonormal basis of the eigenspace associated with $\lambda_{r}$ with multiplicity $l_r$, $r=0,1,\ldots,d$.  Let $\mathbf{x}^H$ denote the conjugate transpose of a column vector $\mathbf{x}$. Then, for each eigenvalue $\lambda_r$ of $L_G$, define
$$
F_{\lambda_r} = \sum\limits_{i=1}^{l_r}\mathbf{x}_i^{(r)} \left(\mathbf{x}_i^{(r)}\right)^H,
$$
which is usually called the \emph{eigenprojector} (or orthogonal projector onto an eigenspace) corresponding to  $\lambda_r$ of $G$. Note that $\sum_{r=0}^dF_{\lambda_r}=I$ (the identity matrix). Then
\begin{equation}
\label{spect1}
L_G=L_G\sum_{r=0}^dF_{\lambda_r} =\sum_{r=0}^d\sum\limits_{i=1}^{l_r}L_G\mathbf{x}_i^{(r)} \left(\mathbf{x}_i^{(r)}\right)^H  =\sum_{r=0}^d\sum\limits_{i=1}^{l_r}\lambda_r\mathbf{x}_i^{(r)} \left(\mathbf{x}_i^{(r)}\right)^H  =\sum_{r=0}^{d}\lambda_rF_{\lambda_r},
\end{equation}
which is called the \emph{spectral decomposition of $L_G$ with respect to the distinct eigenvalues}  (see ``Spectral Theorem for Diagonalizable Matrices'' in \cite[Page 517]{MAALA}). Note that $F_{\lambda_r}^{2}=F_{\lambda_r}$ and $F_{\lambda_r}F_{\lambda_s}=\mathbf{0}$ for $r\neq s$, where $\mathbf{0}$ denotes the zero matrix. So, by (\ref{spect1}), we have
\begin{equation}\label{LSpecDec2}
H_{L_G}(t)=\sum_{k\geq 0}\dfrac{(-\mathrm{i})^{k}L_G^{k}t^{k}}{k!}=\sum_{k\geq 0}\dfrac{(-\mathrm{i})^{k}\left(\sum\limits_{r=0}^{d}\lambda_{r}^{k}F_{\lambda_r}\right)t^{k}}{k!} =\sum_{r=0}^{d}\exp(-\mathrm{i}t\lambda_{r})F_{\lambda_r}.
\end{equation}

The \emph{Laplacian eigenvalue support} of a vertex $u$ in $G$, denoted by $\mathrm{{supp}}_{L_G}(u)$, is the set of all eigenvalues $\lambda$ of $L_G$ such that $F_\lambda\mathbf{e}_u\neq \mathbf{0}$, where $\mathbf{e}_u$ is the characteristic vector corresponding to $u$. Two vertices $u$ and $v$ are \emph{strongly Laplacian cospectral} if $F_\lambda\mathbf{e}_u=\pm F_\lambda\mathbf{e}_v$ for each eigenvalue $\lambda$ of $L_G$.

The following result, which rephrases \cite[Theorem 7.3.1]{Coutinho16}, gives a necessary and sufficient condition for a graph to have LPST.

\begin{theorem}\emph{(see \cite[Theorem 2.1]{Ack})}
\label{Coutinho}
Let $u$ and $v$ be two distinct vertices in a graph $G$. Set $S=\mathrm{supp}_{L_G}(u)$. Then LPST occurs between $u$ and $v$ at time $\tau$ in $G$ if and only if all of the following hold:
\begin{itemize}
\item[\rm (a)]  Vertices $u$ and $v$ are strongly Laplacian cospectral;
 \item[\rm (b)] For each $\lambda\in S$, $\lambda$ is an integer;
  \item[\rm (c)] For each $\lambda\in S$, $\mathbf{e}_u^{\top}F_\lambda \mathbf{e}_v$ is positive if and only if $\lambda/\gcd(S)$ is even, where $\gcd(S)$ denotes the great common divisor of all elements in $S$.
\end{itemize}
Moreover, if these hold, there is a minimum time of LPST given by $$t_0:=\frac{\pi}{\gcd(S)},$$
and $\tau$ is an odd multiple of $t_0$.
\end{theorem}

Similar to the approach Godsil et al. used in \cite{Godsil1},  we need the Kronecker's Approximation Theorem to help us to study the existence of LPGST in vertex complemented coronas.

\begin{theorem}
\emph{(see \cite[Theorem 442]{Hw})}
\label{H-W}
Let $1,\lambda_1,\ldots,\lambda_m$ be linearly independent over $\mathbb{Q}$. Let $\alpha_1,\ldots,\alpha_m$ be arbitrary real numbers, and let $\varepsilon$ be a positive real number. Then there exist  integers $l$ and  $q_1,\ldots, q_m$ such that
\begin{equation}
\label{KroApp}
| l\lambda_k-\alpha_k-q_k|<\varepsilon,
\end{equation}
for each $k=1,\ldots,m$.
\end{theorem}

Whenever we have an inequality of the form $|\alpha-\beta|<\varepsilon$ for arbitrarily small $\varepsilon$, we will write instead $\alpha\approx \beta$ and omit the explicit dependence on $\varepsilon$. For example, (\ref{KroApp}) will be represented as $l\lambda_k-q_k\approx\alpha_k$.

When we apply Kronecker's Approximation Theorem to study the existence of LPGST in vertex complemented coronas, the following result to identify sets of numbers which are linearly independent over the rationals will be involved.

\begin{theorem}
\emph{(see \cite[ Theorem 1a]{Ri}) }
\label{Ri} Let $p_1,p_2,\ldots,p_k$ be distinct positive  primes.
The set $\left\{\sqrt[n]{p_1^{m(1)}\cdots p_k^{m(k)}}: 0\leq m(i)<n,~1\leq i \leq k \right\}$ is linearly independent over the set of rational numbers $\mathbb{Q}$.
\end{theorem}

When $n=2$, Theorem \ref{Ri} immediately implies the following result.

\begin{cor}
\label{Ri1}
The set $\left\{\sqrt{\Delta}: \Delta\text{~is~a~square-free~integer}\right\}$ is linearly independent over the set of rational numbers $\mathbb{Q}$.
\end{cor}

%%%%%%%%%%%%%%%%%%%%%%

\section{Laplacian eigenvalues and eigenprojectors of vertex complemented coronas}\label{Sec:cor}

Let $G$ be a graph with  vertex set $V(G)=\{v_1,\ldots,v_n\}$ and let $\overrightarrow{H}=(H_1,\ldots,H_n)$
be an $n$-tuple of graphs.  Formally, we label the vertex set of $G\tilde{\circ}\overrightarrow{H}$ as follows:
\begin{equation}\label{www}
  V(G\tilde{\circ}\overrightarrow{H})=\left\{(v,0): v\in V(G)\right\} \cup\bigcup_{j=1}^{n}\left\{(v_{j},w):v_{j}\in V(G), w\in V(H_{j})\right\},
\end{equation}
and the adjacency relation
\begin{equation}
	\label{relation}
(v_i,w)\sim(v_j,w') \Longleftrightarrow \left\{
           \begin{array}{lr}
     w=w'=0\text{~and~}  v_i\sim  v_j \text{~in~} G, & \text{or}\\[0.2cm]
   v_i= v_j\text{~and~} w\sim w'  \text{~in~} {H_l},  &   \text{or}\\[0.2cm]
  v_i\not=v_j\text{~and~exactly~one~of~$w$~and~$w'$~is~$0$}.& \end{array}\right.
\end{equation}

%Now we denoted an $n\times n$ matrix $M$ with entry  $m_{ij}=1$ if and only if vertex $v_{i}$ adjacent to all the vertices in $H_{j}$ for $i=1,2,\ldots,n$ and $j=1,2,\ldots,n$.

%By the definition of \emph{vertex complemented corona}, we define $M=J_{n}-I_{n}$ corresponding to graph $G\tilde{\circ}\overrightarrow{H}$.

\noindent\textbf{Notation.} Recall that $\mathbf{j}_m$ denotes the column vector of size $m$ with all entries equal to one. Let $J_{m\times n}$ denote the $m\times n$ matrix with all entries equal to one. In particular, if $m=n$, we simply write $J_{m\times m}$ by $J_{m}$.
Let $\mathbf{e}_i^{n}$ denote the unit vector of size $n$ with the $i$-th entry equal to $1$. If the size $n$ of $\mathbf{e}_i^{n}$ can be easily read from the context, then we omit the superscript and write $\mathbf{e}_i^{n}$ as $\mathbf{e}_i$ for simplicity.

\begin{theorem}
\label{Hou1}
Let $G$ be a connected graph with $n$ vertices, $\overrightarrow{H}=(H_1,H_2,\ldots,H_n)$ an $n$-tuple of graphs with $|V(H_i)|=k\geq 1$, $i=1,2,\ldots,n$. Suppose that $G$ has Laplacian eigenvalues $0=\theta_0<\theta_1<\cdots<\theta_p$ with multiplicities $1=s_0,s_1,\ldots,s_p$. Then the Laplacian eigenvalues of $G\tilde{\circ}\overrightarrow{H}$ are
\begin{itemize}
\item[\rm (a)] $n-1$ with multiplicity $\left(\sum\limits_{i=1}^n s^i_{0}\right)-n$, where $s^i_{0}$ denotes the multiplicity of Laplacian eigenvalue $0$ of $H_i$;
 \item[\rm (b)]  $n-1+\mu$ with multiplicity $s_{\mu}$, where $\mu$ denotes a non-zero Laplacian eigenvalue of $H_i$ with multiplicity $s_{\mu}$, and $\mu$  runs over all non-zero Laplacian eigenvalue of $H_i$, for $i=1,2,\ldots,n$;

  \item[\rm (c)]  $\frac{1}{2}\left((n-1)(1+k)+\theta_j\pm\sqrt{((n-1)(k-1)+\theta_j)^2+4k}\right)$ with multiplicity $s_j$, for $j=1, 2, \ldots, p$;
      \item[\rm (d)] $(n-1)(1+k)$ with multiplicity $1$;
 \item[\rm (e)]  $0$ with multiplicity $1$.

 \end{itemize}
\end{theorem}
 \begin{proof}
Define $M=J_n-I_n$. The Laplacian matrix of $G\tilde{\circ}\overrightarrow{H}$ is given by
\begin{equation}\label{LmatrixGH}
	L_{G\tilde{\circ}\overrightarrow{H}}=
	\left(
	\begin{array}{cc}
	L_G+(n-1)kI_{n}& -M\otimes \mathbf{j}_{k}^{\top}\\ [0.2cm]
	-M^{\top}\otimes \mathbf{j}_{k}& \sum\limits_{i=1}^{n} \left(\mathbf{e}^n_i(\mathbf{e}^n_i)^\top\otimes (L _{H_i}+(n-1)I_{k})\right)
	\end{array}
	\right).
	\end{equation}
By Lemma \ref{tian6},  the Laplacian characteristic polynomial of $G\tilde{\circ}\overrightarrow{H}$ is
 \begin{align*}
   \phi(L_{G\tilde{\circ}\overrightarrow{H}};x)=&\det\left(
	\begin{array}{cc}
	(x-(n-1)k)I_{n}-L_G& M\otimes \mathbf{j}_{k}^{\top}\\ [0.2cm]
	M^{\top}\otimes \mathbf{j}_{k}& \sum\limits_{i=1}^{n} \left( \mathbf{e}^n_i(\mathbf{e}^n_i)^\top\otimes ((x-(n-1))I_{k}-L _{H_i})\right)
	\end{array}
	\right)\\
=&\det(N)\det(S),
 \end{align*}
where
\begin{equation*}
  N=\sum\limits_{i=1}^{n} \left(\mathbf{e}^n_i(\mathbf{e}^n_i)^\top\otimes ((x-(n-1))I_{k}-L _{H_i})\right),
\end{equation*}
and
\begin{equation*}
  S=(x-(n-1)k)I_{n}-L _{G}-{(M\otimes \mathbf{j}_{k}^{\top})N^{-1}(M^{\top}\otimes \mathbf{j}_{k})}.
\end{equation*}
By Lemma \ref{tian12}, we have
\begin{equation*}\label{eeee}
 (M\otimes \mathbf{j}_{k}^{\top})N^{-1}(M^{\top}\otimes \mathbf{j}_{k})=\frac{k}{x-n+1}MM^{\top}=\frac{k}{x-n+1}(I_{n}+(n-2)J_{n}).
\end{equation*}
So, by Lemmas \ref{tian12} and \ref{liuzhang1}, we have
\begin{align*}
\det(S)=& \det\left(\left(x-(n-1)k-\frac{k}{x-n+1}\right)I_{n}-L_{G}-\frac{(n-2)k}{x-n+1} J_{n}\right)\\
=&\left(1-\frac{(n-2)k}{x-n+1}\cdot\Gamma_{L_G}\left(x-(n-1)k-\frac{k}{x-n+1}\right)\right)\\
&\cdot \det\left(\left(x-(n-1)k-\frac{k}{x-n+1}\right)I_{n}-L_{G}\right)\\
=&(x-n+1)^{-n}\cdot\left(1-\frac{(n-2)k}{x-n+1}\cdot \frac{n}{x-(n-1)k-\frac{k}{x-n+1}}\right)\\
&\cdot \det\left(\left((x-(n-1)k)(x-n+1)-k\right)I_{n}-(x-n+1)L_{G}\right)\\[0.2cm]
=&(x-n+1)^{-n}\cdot \frac{(x-(n-1)k)(x-n+1)-k-k(n-2)n}{(x-(n-1)k)(x-n+1)-k}\\
   &\cdot\prod\limits_{j=0}^{p}\left((x-(n-1)k)(x-n+1)-k-(x-n+1)\theta_{j}\right)^{s_j}\\
=&(x-n+1)^{-n}\cdot x(x-(n-1)(1+k))\\
   &\cdot\prod\limits_{j=1}^{p}\left((x^{2}-((n-1)(1+k)+\theta_{j})x+(n-1)^{2}k-k+(n-1)\theta_{j}\right)^{s_j}.
\end{align*}
Note that
\begin{align*}
\det(N)=&\prod_{i=1}^{n} \det((x-(n-1))I_{k}-L _{H_i}).
\end{align*}
Therefore, the required result follows from  $\phi(L_{G\tilde{\circ}\overrightarrow{H}};x)
=\det(N)\det(S)$.
\qed
\end{proof}

Now, we compute the Laplacian  eigenprojectors of $G\tilde{\circ}\overrightarrow{H}$.

\begin{theorem}\label{eigenprojector}
Let $G$ be a connected graph with $n$ vertices, and $\overrightarrow{H}=(H_1,H_2,\ldots,H_n)$ an $n$-tuple of graphs with $|V(H_i)|=k\geq 1$, $i=1,2,\ldots,n$. Then  the Laplacian  eigenprojectors of $G\tilde{\circ}\overrightarrow{H}$ are stated as follows:
\begin{itemize}
   \item[\rm (a)]  If there exists $H_l\in \{H_1,\ldots,H_n\}$ that is disconnected,  then $n-1$ is a Laplacian eigenvalue of $G\tilde{\circ}\overrightarrow{H}$ with the eigenprojector
          \begin{equation}\label{F013}
          F_{n-1}=
	\left(
	\begin{array}{cc}
	\mathbf{0}& \mathbf{0} \\ [0.2cm]
	\mathbf{0}& \sum\limits_{l=1}^{n} \left(\mathbf{e}^{n}_l(\mathbf{e}^{n}_l)^\top\right)\otimes \left(F_0(H_l)-\frac{1}{k}J_k\right)
	\end{array}
	\right),	
\end{equation}
where $F_0(H_l)$ denotes the eigenprojector corresponding to the Laplacian eigenvalue $0$ of $H_l$.
     \item[\rm (b)]If $\mu$ is a nonzero Laplacian eigenvalue of some $H_l\in \{H_1,\ldots,H_m\}$, then $n-1+\mu$ is a Laplacian eigenvalue of $G\tilde{\circ}\overrightarrow{H}$ with the eigenprojector
        \begin{equation}\label{F13}
	F_{n-1+\mu}=
	\left(
	\begin{array}{cc}
	\mathbf{0}& \mathbf{0} \\ [0.2cm]
	\mathbf{0}& \sum\limits_{l=1}^{n} \left(\mathbf{e}^{n}_l(\mathbf{e}^{n}_l)^\top\right)\otimes F_\mu(H_l)
	\end{array}
	\right),
	\end{equation}
where $F_\mu(H_l)$ denotes the eigenprojector corresponding to the Laplacian eigenvalue $\mu$ of $H_l$. Here, we assume that $F_\mu(H_l)=0$ if $\mu$ is not a Laplacian eigenvalue of $H_l$.
  \item[\rm (c)] If $\theta\not=0$ is a Laplacian eigenvalue of $G$, then $$\theta_{\pm}:=\frac{1}{2}\left((n-1)(1+k) +\theta\pm\sqrt{((n-1)(k-1)+\theta)^2+4k}\right)$$ are Laplacian eigenvalues of $G\tilde{\circ}\overrightarrow{H}$ with the eigenprojectors
     \begin{align}\label{F14} F_{\theta_{\pm}} =&\frac{(n-1-\theta_{\pm})^{2}}{(n-1-\theta_{\pm})^{2}+k}
 \left(
	\begin{array}{cc}
	F_{\theta}(G)& \frac{1}{n-1-\theta_{\pm}} \left(F_{\theta}(G)M\right)\otimes \mathbf{j}_{k}^\top\\ [0.3cm]
	\frac{1}{n-1-\theta_{\pm}}( F_{\theta}(G)M)^\top\otimes \mathbf{j}_{k}& \frac{1}{(n-1-\theta_{\pm})^2}\left(M^\top F_{\theta}(G)M\right) \otimes J_{k}
	\end{array}
	\right),
	\end{align}
where $M=J_n-I_n$ and $F_\theta(G)$ denotes the eigenprojector corresponding to eigenvalue $\theta\neq0$ of $G$.
\item[\rm (d)]  $(n-1)(1+k)$ is a Laplacian eigenvalue of $G\tilde{\circ}\overrightarrow{H}$ with the eigenprojector
     \begin{equation}\label{F15}
	F_{(n-1)(1+k)}=\frac{1}{nk(k+1)}
	\left(
	\begin{array}{cc}
	k^{2}J_{n} & -k J_{n}\otimes \mathbf j_{k}^{\top}\\ [0.2cm]
	-k J_{n}\otimes \mathbf j_{k} &   J_{n}\otimes J_{k}\\ [0.2cm]
	
	\end{array}
	\right),
	\end{equation}
\item[\rm (e)]
 $0$ is a Laplacian eigenvalue of $G\tilde{\circ}\overrightarrow{H}$ with the eigenprojector
     \begin{equation}\label{F16}
	F_{0}=\frac{1}{n(k+1)}J_{n(k+1)}.
	 \end{equation}
\end{itemize}
Therefore, the spectral decomposition of  $L_{G\tilde{\circ}\overrightarrow{H}}$ is given by
\begin{equation}\label{S1}
	L_{G\tilde{\circ}\overrightarrow{H}}=\sum_{\theta \in \mathrm{Spec}_L(G)\setminus\left\{0\right\}}\sum_{\pm} \theta_{\pm}F_{\theta_{\pm}}+\sum_{\mu}(n-1+\mu)F_{n-1+\mu}+(n-1)(1+k)F_{(n-1)(1+k)},
	\end{equation}
where the sum over $\mu$ is over all eigenvalues of the graphs $H_l$, for $l=1,\ldots,n$.
\end{theorem}

\begin{proof}
The proofs of (a), (b), (c), (d) and (e) consist of Claims 1--3.
\smallskip

\noindent\textbf{Claim 1:} $\mathbf{X}$,  $\mathbf{Y}$ and $\mathbf{Z_{\pm}}$ defined below are the eigenvectors of $G\tilde{\circ}\overrightarrow{H}$ corresponding to the Laplacian eigenvalues $n-1+\mu$, $(n-1)(1+k)$ and $\theta_{\pm}$.
\smallskip

\noindent\emph{Proof of Claim 1.}
Let $H_l$ be  one of graphs in $\overrightarrow{H}$. Note that $0$ is always a Laplacian eigenvalue of $H_l$ with eigenvector $\mathbf{j}_{k}$, and the multiplicity of $0$ is equal to the number of connected components of $H_l$. Note also that $F_0(H_l)=\frac{1}{k}J_k$ if and only if $H_l$ is connected. Suppose that  $\mathbf{x}\perp \mathbf{j}_{k}$ is a normalized eigenvector of $L_{H_l}$ corresponding to the Laplacian eigenvalue $\mu$ of $H_l$ (Here, $\mu$ may be equal to $0$ if $H_l$ is disconnected). Then, we have a $(n+nk)$-dimension vector $\mathbf{X}:=\left(\begin{array}{c}
\mathbf{0}_{n\times 1} \\[0.2cm]
\mathbf{e}_l^n\otimes\mathbf{x}
\end{array}
\right)$ such that
\begin{align}\label{Sp1}
 L_{G\tilde{\circ}\overrightarrow{H}}\mathbf{X}
 &=(n-1+\mu)\mathbf{X},
\end{align}
where $\mathbf{0}_{s\times t}$ denotes $s\times t$ matrix with all entries equal to $0$. Then $n-1+\mu$ is a Laplacian eigenvalue of $L_{G\tilde{\circ}\overrightarrow{H}}$ with the eigenvector $\mathbf{X}$ (Here, $n-1$ is also a Laplacian eigenvalue of $L_{G\tilde{\circ}\overrightarrow{H}}$ with a specified eigenvector $\mathbf{X}$  if $H_l$ is disconnected).

Suppose that $\mathbf{z}\perp \mathbf{j}_{n}$ is a normalized eigenvector of $L_{G}$  corresponding to the Laplacian eigenvalue $\theta\not=0$.
Define a vector $\mathbf{Z}_{\pm}:=\left(\begin{array}{c}
                                                                 \mathbf{z}\\[0.2cm]
                                                                 \frac{1}{n-1-\theta_{\pm}}(M^\top\mathbf{z}) \otimes \mathbf{j}_{k}
                                                               \end{array}
\right)$. Notice that
$$ MM^\top=(J_n-I_n)(J_n-I_n)^\top=I_{n}+(n-2) J_{n},$$
where $M=J_n-I_n$. Note also that $\mathbf{z}$ can be regarded as $\mathbf{z}\otimes 1$. Then, by (\ref{LmatrixGH}),
\begin{align}\label{Sp2}
 &L_{G\tilde{\circ}\overrightarrow{H}}\mathbf{Z}_{\pm}\nonumber\\
= &\left(
	\begin{array}{cc}
	L_G+(n-1)kI_{n}& -M\otimes \mathbf{j}_{k}^{\top}\\ [0.2cm]
	-M^{\top}\otimes \mathbf{j}_{k}& \sum\limits_{i=1}^{n} \left(\mathbf{e}^n_i(\mathbf{e}^n_i)^\top\otimes (L _{H_i}+(n-1)I_{k})\right)
	\end{array}
	\right)
\left(\begin{array}{c}
\mathbf{z}\\[0.2cm]
\frac{1}{n-1-\theta_{\pm}}(M^\top\mathbf{z}) \otimes \mathbf{j}_{k}
\end{array}
\right)\nonumber\\
= & \left(\begin{array}{c}
 (\theta+(n-1)k)\mathbf{z} - \frac{1}{n-1-\theta_{\pm}}(MM^\top\mathbf{z}) \otimes \mathbf{j}_{k}^{\top}\mathbf{j}_{k} \\[0.2cm]
 -(M^\top\mathbf{z}) \otimes\mathbf{j}_{k} +\frac{n-1}{n-1-\theta_{\pm}}(M^\top\mathbf{z}) \otimes \mathbf{j}_{k}
 \end{array}
\right)\nonumber\\
= & \left(\begin{array}{c}
 (\theta+(n-1)k)\mathbf{z} - \frac{k}{n-1-\theta_{\pm}}\mathbf{z} \\[0.2cm]
 \frac{\theta_{\pm}}{n-1-\theta_{\pm}}(M^\top\mathbf{z}) \otimes \mathbf{j}_{k}
 \end{array}
\right)\nonumber\\
=& \theta_{\pm}\mathbf{Z}_{\pm}.
\end{align}
Thus $\theta_{\pm}$ are eigenvalues of $L_{G\tilde{\circ}\overrightarrow{H}}$ with the eigenvector $\mathbf{Z}_{\pm}$ for each Laplacian eigenvalue $\theta\not=0$ of $L_G$.

Define a vector $\mathbf{Y}:=\left(\begin{array}{c}
                                    -k\mathbf j_{n} \\
                                    \mathbf j_{n}\otimes \mathbf j_{k}
                                  \end{array}\right)$. It is easy to verify that
     \begin{align}\label{Sp4}
 L_{G\tilde{\circ}\overrightarrow{H}}\mathbf{Y}
 &= (n-1)(1+k)\mathbf{Y}.
\end{align}
Thus $(n-1)(1+k)$ is a Laplacian eigenvalue of $L_{G\tilde{\circ}\overrightarrow{H}}$ with the eigenvector $\mathbf{Y}$.
%\end{itemize}
%%  $2n_1$ eigenvalues and corresponding eigenvectors of $L(G\diamond \overrightarrow{H})$.
%
%%Now, we have obtained $m(n_2-1)+2n_1$ eigenvalues of $L_{G\diamond \overrightarrow{H}}$. Note that $L_{G\diamond \overrightarrow{H}}$ has $n_1+mn_2$ eigenvalues. In the following, we construct the remaining $n_1+mn_2-(m(n_2-1)+2n_1)=m-n_1$ eigenvalues by using the solutions of $R_G\mathbf{y}=\mathbf{0}$.
%
%
%

\smallskip
\noindent\textbf{Claim 2:} All $\mathbf{X}$'s, $\mathbf{Y}$'s and $\mathbf{Z}_{\pm}$'s are orthogonal eigenvectors of  $L_{G\tilde{\circ}\overrightarrow{H}}$.
\smallskip

\noindent\emph{Proof of Claim 2.} Recall that $\mathbf{x}\perp \mathbf{j}_{k}$. Then $\mathbf{X}\perp \mathbf{Y}$, $\mathbf{X}\perp \mathbf{Z}_{\pm}$. Recall also that $\mathbf{z}\perp \mathbf{j}_{n}$. Then $\mathbf{Z}_{\pm}\perp \mathbf{Y}$.
Let $\mathbf{z}$ and $\mathbf{z}'$ be orthogonal  eigenvectors of $L_{G}$ corresponding to nonzero eigenvalues $\theta$ and $\theta'$, respectively (Here, $\theta$ and $\theta'$ may be equal). Then
\begin{align*}
(M^\top\mathbf{z})^\top(M^\top\mathbf{z}')=&\mathbf{z}^\top
MM^\top\mathbf{z}'=\mathbf{z}^\top(I_{n}+(n-2) J_{n})\mathbf{z}'
 =\mathbf{z}^\top\mathbf{z}'
 =0.
\end{align*}
Consider $\mathbf{Z}_{\pm}=\left(\begin{array}{c}
                                                                 \mathbf{z}\\[0.2cm]
                                                                 \frac{1}{n-1-\theta_{\pm}}(M^\top\mathbf{z}) \otimes \mathbf{j}_{k}
                                                               \end{array}
\right)$ and $\mathbf{Z}'_{\pm}=\left(\begin{array}{c}
                                                                 \mathbf{z}'\\[0.2cm]
                                                                 \frac{1}{n-1-\theta_{\pm}'}(M^\top\mathbf{z}') \otimes \mathbf{j}_{k}
                                                               \end{array}
\right)$.
Note that
\begin{align*}
\mathbf{Z}_{\pm}^\top\mathbf{Z}'_{\pm} &=\mathbf{z}^\top\mathbf{z}'+
\frac{(M^\top\mathbf{z})^\top(M^\top\mathbf{z}')\otimes \mathbf{j}^\top_{k}\mathbf{j}_{k}}{(n-1-\theta_{\pm})(n-1-\theta_{\pm}')}=0.
\end{align*}
Then $\mathbf{Z}_{\pm}\perp \mathbf{Z}'_{\pm}$.

For $\mathbf{Z}_{+},\, \mathbf{Z}_{-}$, note that
$$
(M^\top\mathbf{z})^\top(M^\top\mathbf{z})=\mathbf{z}^\top
MM^\top\mathbf{z}=\mathbf{z}^\top(I_{n}+(n-2) J_{n})\mathbf{z}= \mathbf{z}^\top\mathbf{z}=1,
$$
and
\begin{align*}
(n-1-\theta_{+})(n-1-\theta_{-})=& -k.
\end{align*}
So, we have
\begin{align*}
\mathbf{Z}_{+}^\top\mathbf{Z}_{-}&=\mathbf{z}^\top\mathbf{z}+
\frac{(M^\top\mathbf{z})^\top(M^\top\mathbf{z})\otimes \mathbf{j}^\top_{k}\mathbf{j}_{k}}{(n-1-\theta_{+})(n-1-\theta_{-})} =1+\frac{k}{-k}=0,
\end{align*}
which implies that $\mathbf{Z}_{+}\perp \mathbf{Z}_{-}$.

\smallskip

\noindent\textbf{Claim 3:} (\ref{F013}), (\ref{F13}), (\ref{F14}), (\ref{F15}) and (\ref{F16}) are eigenprojectors of $L_{G\tilde{\circ}\overrightarrow{H}}$ corresponding to Laplacian eigenvalues $n-1$, $n-1+\mu$, $\theta_{\pm}$, $(n-1)(1+k)$  and $0$.
\smallskip

\noindent\emph{Proof of Claim 3.} If there exists $H_l\in \{H_1,\ldots,H_m\}$ that is disconnected, by (\ref{Sp1}) and (\ref{Sp4}), it is easy to verify that (\ref{F013}), (\ref{F13}) and (\ref{F15}) are the eigenprojectors corresponding to the Laplacian eigenvalues $n-1$, $n-1+\mu$ and $(n-1)(1+k)$.  Note that $\mathbf{j}_{n(k+1)}$ is the eigenvector of the Laplacian eigenvalue $0$ of $G\tilde{\circ}\overrightarrow{H}$. Then, one can easily verify that (\ref{F16}) is the eigenprojector of $L_{G\tilde{\circ}\overrightarrow{H}}$ corresponding to $0$.

Let $\theta\not=0$ be a Laplacian eigenvalue of $L_{G}$.
Suppose that $\left\{\mathbf{z}^{(1)},\ldots, \mathbf{z}^{(s)}\right\}$ is an orthonormal basis of the eigenspace corresponding to $\theta$. Set $\mathbf{Z}^{(i)}_{\pm}=\left(\begin{array}{c}
                                                                \mathbf{z}^{(i)}\\[0.2cm]
                                                                 \frac{1}{n-1 -\theta_{\pm}}(M^\top\mathbf{z}^{(i)}) \otimes \mathbf{j}_{k}
                                                               \end{array}
\right)$. Note that
$\left(M^\top\mathbf{z}^{(i)}\right)^\top\left(M^\top\mathbf{z}^{(i)}\right)=1$. Then $\left\|\mathbf{Z}^{(i)}_{\pm}\right\|^2=1+\frac{k}{(n-1-\theta_{\pm})^2}$. Let
$F_{\theta}(G)=\sum\limits_{i=1}^s\left(\mathbf{z}^{(i)}\right)(\mathbf{z}^{(i)})^\top$ be the eigenprojector corresponding to $\theta$ of $G$. Then eigenprojectors of $L_{G\tilde{\circ}\overrightarrow{H}}$ corresponding to $\theta_{\pm}$  are given as follows:
\begin{align*}
&F_{\theta_{\pm}}(G\tilde{\circ}\overrightarrow{H}) \\ =&\frac{(n-1-\theta_{\pm})^{2}}{(n-1-\theta_{\pm})^{2}+k}\sum_{i=1}^{s} \mathbf{Z}^{(i)}_{\pm} \left(\mathbf{Z}^{(i)}_{\pm}\right)^\top\\ =&\frac{(n-1-\theta_{\pm})^{2}}{(n-1-\theta_{\pm})^{2}+k}   \\
&\cdot
\small{
\left(\begin{array}{cc}
\sum\limits_{i=1}^{s}
\left(\mathbf{z}^{(i)}\right)\left(\mathbf{z}^{(i)}\right)^\top   & \frac{1}{n-1-\theta_{\pm}} \left(\left(\sum\limits_{i=1}^{s}
\left(\mathbf{z}^{(i)}\right)\left(\mathbf{z}^{(i)}\right)^\top\right)M\right)\otimes \mathbf{j}^\top_{k} \\ [0.45cm]
	\frac{1}{n-1-\theta_{\pm}} \left(\left(\sum\limits_{i=1}^{s}
\left(\mathbf{z}^{(i)}\right)\left(\mathbf{z}^{(i)}\right)^\top\right)M\right)^\top\otimes \mathbf{j}_{k} &  \frac{1}{(n-1-\theta_{\pm})^2} \left(\sum\limits_{i=1}^{s}\left( M^\top\mathbf{z}^{(i)}\right)\left(\left(\mathbf{z}^{(i)}\right)^\top M\right) \right)\otimes \mathbf{j}_{k}\mathbf{j}^\top_{k}
	\end{array}
	\right)
}\\[0.3cm]
=&\frac{(n-1-\theta_{\pm})^{2}}{(n-1-\theta_{\pm})^{2}+k}
	\left(
\begin{array}{cc}
F_{\theta}(G)& \frac{1}{n-1-\theta_{\pm}} \left(F_{\theta}(G)M\right)\otimes \mathbf{j}_{k}^\top\\ [0.3cm]
	\frac{1}{n-1-\theta_{\pm}}( F_{\theta}(G)M)^\top\otimes \mathbf{j}_{k}& \frac{1}{(n-1-\theta_{\pm})^2}\left(M^\top F_{\theta}(G)M\right) \otimes J_{k}
	\end{array}
	\right),
\end{align*}
yielding (\ref{F14}).

At last, one can verify that (\ref{S1}) is  the spectral decomposition of $L_{G\tilde{\circ}\overrightarrow{H}}$. %This completes the proof.
\qed\end{proof}

%In Proposition \ref{eigenprojector}, we computed the eigenvalues and eigenprojectors of $G\diamond \overrightarrow{H}$ when $G$ is a regular non-bipartite graph and $\overrightarrow{H}=(H_1,\ldots,H_m)$ satisfies that $|V(H_1)|=|V(H_2)|=\cdots=|V(H_m)|$. Similarly, when $G$ is a regular bipartite graph and $\overrightarrow{H}=(H_1,\ldots,H_m)$ satisfies that $|V(H_1)|=|V(H_2)|=\cdots=|V(H_m)|$, we obtain the following result.
%
%%Similarly, we obtain the following result.
%
%\begin{prop}\label{eigen}
%Let $G$ be an $r$-regular bipartite graph with $n_1$ vertices, $m$ edges and $r\ge2$, and let $\overrightarrow{H}=(H_1,\ldots,H_m)$ satisfy that $|V(H_1)|=|V(H_2)|=\cdots=|V(H_m)|=n_2\geq 1$. Suppose that $\{\mathbf{y}_1,\mathbf{y}_2,\ldots,\mathbf{y}_{m-n_1+1}\}$ is a maximal set of independent solution vectors of the linear system $R_G\mathbf{y}=\mathbf{0}$. Then the eigenvalues and eigenprojectors $G\diamond \overrightarrow{H}$ are as follows:
%
%\end{prop}
%
%\begin{proof}
%The proof follows by refining the arguments used to prove Proposition \ref{eigenprojector}, and hence we omit the details.
%\qed\end{proof}

By Theorem \ref{eigenprojector}, we have the following result.

\begin{prop}\label{the_spectral}
Let $G$ be a connected graph with $n$ vertices, and $\overrightarrow{H}=(H_1,H_2,\ldots,H_n)$ an $n$-tuple of graphs  with $|V(H_i)|=k\geq 1$, $i=1,2,\ldots,n$. Let $u$ and $v$ be two distinct vertices of $G$. Then
\begin{align*}\label{FGH}
 &\mathbf{e}^\top_{(u,0)}\exp(-\mathrm{i}t L_{G\tilde{\circ}\overrightarrow{H} })\mathbf{e}_{(v,0)}\nonumber\\[0.2cm]
 =&e^{-\mathrm{i}t\frac{(n-1)(1+k)}{2}}\\[0.2cm]
 &\cdot\left(\sum_{\theta \in \mathrm{Spec}_L(G)\setminus\left\{0\right\}}
e^{-\mathrm{i}t\frac{\theta}{2}}(\mathbf{e}^{n}_u)^\top F_\theta(G)\mathbf{e}^{n}_v
\left(\cos\left(\frac{\Delta_\theta t}{2}\right)+
\mathrm{i}\frac{(n-1)(1-k)-\theta}{\Delta_\theta}
\sin\left(\frac{\Delta_\theta t}{2}\right)\right)\right)\\[0.2cm]
&+\frac{k}{n(k+1)}e^{-\mathrm{i}t(n-1)(1+k)}+\frac{1}{n(k+1)},
\end{align*}
where $\Delta_\theta=
     \sqrt{((n-1)(k-1)+\theta)^2+4k}$, for each Laplacian eigenvalue $\theta\neq0$ of $G$.
\end{prop}

\begin{proof}
Recall that
$\theta_{\pm}=\frac{1}{2}((n-1)(1+k)+\theta\pm\Delta_\theta)$
for each eigenvalue $\theta\neq0$ of $L_G$. By Theorem \ref{eigenprojector} and Equation (\ref{LSpecDec2}), we have
 \begin{align*}
	&\mathbf{e}^\top_{(u,0)}\exp(-\mathrm{i}t L_{G\tilde{\circ}\overrightarrow{H}})\mathbf{e}_{(v,0)}\nonumber\\[0.2cm]
=&\sum_{\theta\in \mathrm{Spec}_L(G)\setminus\left\{0\right\}}e^{-\mathrm{i}t\frac{(n-1)(1+k)+\theta}{2}}
(\mathbf{e}^{n}_u)^\top F_\theta(G)\mathbf{e}^{n}_v\left(\sum_{\pm}
e^{\mp\mathrm{i}\frac{\Delta_\theta t}{2}}\frac{(n-1-\theta_{\pm})^2}{(n-1-\theta_{\pm})^2+k}\right)\\[0.2cm]
&+\frac{k}{n(k+1)}e^{-\mathrm{i}t(n-1)(1+k)}+\frac{1}{n(k+1)}.
 \end{align*}
Note that
$$(n-1-\theta_+)(n-1-\theta_-)=-k,$$
$$(n-1-\theta_+)^2+(n-1-\theta_-)^2=(\theta+(n-1)(k-1))^{2}+2k,$$
and
$$(n-1-\theta_-)^2-(n-1-\theta_+)^2=\Delta_\theta((n-1)(1-k)-\theta).$$
Thus
\begin{eqnarray*}
\left(\sum_{\pm}
e^{\mp\mathrm{i}\frac{\Delta_\theta t}{2}}\frac{(n-1-\theta_{\pm})^2}{(n-1-\theta_{\pm})^2+k}\right)
=\cos\left(\frac{\Delta_\theta t}{2}\right)+
\mathrm{i}\frac{(n-1)(1-k)-\theta}{\Delta_\theta}
\sin\left(\frac{\Delta_\theta t}{2}\right).
 \end{eqnarray*}
This completes the proof. \qed
\end{proof}
%
%
%
%%%%%%%%%%%%%%%%%%%%%%%
%
%
%
\section{LPST in vertex complemented coronas}
\label{Sec:main}

In this section, we explore the conditions under which the vertex complemented corona $G\tilde{\circ}\overrightarrow{H}$ can have LPST. We use the label of the vertex set of $G\tilde{\circ}\overrightarrow{H}$ as in (\ref{www}). For a vertex $(v_{j},w)$ with $v_j\in V(G)$ and $w\in V(H_{l})$,  for $j=1, \ldots, n$,  denote by
 \begin{equation}\label{VecLabel-1}
\mathbf{e}_{(v_{j},w)}:=\left(\begin{array}{c}
\mathbf{0}_{n\times 1} \\[0.2cm]
\mathbf{e}_j^n\otimes\mathbf{e}_w^{k}
\end{array}
\right)
 \end{equation}
the unit vector of size $n(k+1)$ with the $(n+(j-1)k+w)$-th entry equal to $1$. For a vertex $(v,0)$ with $v\in V(G)$,  denote by
 \begin{equation}\label{VecLabel-2}
\mathbf{e}_{(v,0)}:=\mathbf{e}_v^{n(k+1)}
 \end{equation}
the unit vector of size $n(k+1)$ with the $v$-th entry equal to $1$.

\begin{lemma}
\label{spectralF}
Let $G$ be a connected graph with $n$ vertices. Suppose that $F_{\theta_0}(G),\ldots,F_{\theta_p}(G)$ are the eigenprojectors corresponding to Laplacian eigenvalues $0=\theta_0,\ldots,\theta_p$ of $G$. Define $M=J_n-I_n$. Then for any $j\in\{1,2,\ldots,n\}$,  there exists some $\theta\in\{\theta_1,\ldots,\theta_p\}$ such that $(F_\theta(G) M)\mathbf{e}^{n}_{j}\neq\mathbf{0}$.
\end{lemma}

\begin{proof}
Note that $\sum\limits_{r=0}^pF_{\theta_r}(G)=I_n$ and $F_{0}(G)=\frac{1}{n}J_{n}$. Then
$$\sum_{r=1}^pF_{\theta_r}(G)=I_n-F_{0}(G)=I_n-\frac{1}{n}J_{n},$$
which implies that $\left(\sum\limits_{r=1}^pF_{\theta_r}(G)M\right)\mathbf{e}^{n}_{j} =(\frac{1}{n}J_{n}-I_n)\mathbf{e}^{n}_{j}\neq \mathbf{0}$, that is, there exists some $\theta\in\{\theta_1,\ldots,\theta_p\}$ such that $(F_\theta(G) M)\mathbf{e}^{n}_{j}\neq\mathbf{0}$. \qed
\end{proof}

\begin{theorem}
\label{nonbip1}
 Let $G$ be a connected graph with $n$ vertices, and $\overrightarrow{H}=(H_1,\ldots,H_n)$ an $n$-tuple of graphs with $|V(H_i)|=k\geq 1$, $i=1,2,\ldots,n$. Then there is no LPST in $G\tilde{\circ}\overrightarrow{H}$.
\end{theorem}

\begin{proof}
Let $(v_{j},w)$ be a vertex of $G\tilde{\circ}\overrightarrow{H}$. By Theorem \ref{Coutinho}, in order to prove that $G\tilde{\circ}\overrightarrow{H}$ does not have LPST,  it suffices to prove that there exists a non-integer eigenvalue in the Laplacian eigenvalue support of $(v_{j},w)$.  Consider the following cases.

\noindent\emph{Case 1.} $w=0$. Since $G$ is connected on at least two vertices, there exists a positive Laplacian eigenvalue $\theta$ in the Laplacian eigenvalue support of $v_{j}$, that is, $F_\theta(G)\mathbf{e}_{j}^{n}\neq \mathbf{0}$, where $F_{\theta}(G)$ is the eigenprojector corresponding to the Laplacian eigenvalue $\theta$ of $G$.
Then, by Theorem \ref{eigenprojector} (c) and (\ref{VecLabel-2}), we have $F_{\theta_\pm } \mathbf{e}_{(v_{j},0)}\neq \mathbf{0}$, which means that both
$$\theta_{\pm}=\frac{1}{2}\left((n-1)(1+k)+\theta\pm\sqrt{((n-1)(k-1)+\theta)^2+4k}\right)$$
are in the Laplacian eigenvalue support of $(v_{j},0)$. Suppose, towards contradiction, that both $\theta_{\pm}$ are integers. Then both
$$(n-1)(1+k)+\theta=\theta_++\theta_-$$
and
$$\sqrt{((n-1)(k-1)+\theta)^2+4k}=\theta_+-\theta_-$$
are integers, implying that $\theta$ is an integer, and $((n-1)(k-1)+\theta)^2+4k$ is a perfect square. Note that $4k$ is even. Then the parity of the square $((n-1)(k-1)+\theta)^2+4k$ must be the same as $((n-1)(k-1)+\theta)^2$. Since $4k\geq4$, we assume that
$$((n-1)(k-1)+\theta)^2+4k= ((n-1)(k-1)+\theta+2a)^2,~(a\geq 1),$$
that is,
$$\theta= -a+\frac{k}{a}-(n-1)(k-1).$$
Then, if $a=1$, $\theta=(k-1)(2-n)\leq0$. Since the function $\theta=-a+\frac{k}{a}-(n-1)(k-1)$ decreases as $a$ decreases, then $\theta\leq0$ if $a\geq1$, a contradiction to $\theta>0$. Thus, one of $\theta_{\pm}$ is not an integer.

\noindent\emph{Case 2.} $w\neq 0$. Lemma \ref{spectralF} implies that, for any $j\in \{1,\ldots,n\}$, there exists some Laplacian eigenvalue $\theta\neq0$ such that $(F_\theta(G) M)\mathbf{e}^{n}_{j}\neq \mathbf{0}$. Then, by Proposition \ref{eigenprojector} (c) and (\ref{VecLabel-1}), we have
\begin{align*}
F_{\theta_\pm }\mathbf{e}_{(v_{j},w)} &=\frac{(n-1-\theta_{\pm})^{2}}{(n-1-\theta_{\pm})^{2}+k}
 \left(
	\begin{array}{cc}
	\frac{1}{n-1-\theta_{\pm}} \left(\left(F_{\theta}(G)M\right)\otimes \mathbf{j}_{k}^\top\right)(\mathbf{e}_j^n\otimes\mathbf{e}_w^{k})\\ [0.3cm]
	 \frac{1}{(n-1-\theta_{\pm})^2}\left(\left(M^\top F_{\theta}(G)M\right) \otimes J_{k}\right)(\mathbf{e}_j^n\otimes\mathbf{e}_w^{k})
	\end{array}
	\right)\\[0.2cm]
&=\frac{(n-1-\theta_{\pm})^{2}}{(n-1-\theta_{\pm})^{2}+k}
 \left(
	\begin{array}{cc}
	\frac{1}{n-1-\theta_{\pm}} (F_{\theta}(G)M)\mathbf{e}_j^n\otimes (\mathbf{j}_{k}^\top\mathbf{e}_w^{k})\\ [0.3cm]
	 \frac{1}{(n-1-\theta_{\pm})^2}(M^\top F_{\theta}(G)M)\mathbf{e}_j^n \otimes (J_{k}\mathbf{e}_w^{k})
	\end{array}
	\right)\\[0.2cm]
 &\neq\mathbf{0},
\end{align*}
which means that both
$$\theta_{\pm}=\frac{1}{2}\left((n-1)(1+k)+\theta\pm\sqrt{((n-1)(k-1)+\theta)^2+4k}\right)$$
are in the Laplacian eigenvalue support of $(v_{j},w)$. By the same argument as in Case 1, one of $\theta_{\pm}$ is not an integer.

This completes the proof. \qed
\end{proof}

\section{LPGST in vertex complemented coronas}
\label{Sec:main2}

In this section, we study the existence of LPGST in vertex complemented coronas.

\begin{theorem}
\label{pregoo-2}
Let $G$ be a connected graph with $n$ vertices and $\overrightarrow{H}=(H_1,\ldots,H_n)$ an $n$-tuple of graphs with $|V(H_i)|=k\geq 1$, $i=1,2,\ldots,n$. Suppose that $G$ has LPST between vertices $u$ and $v$, and let $2^e$ be the greatest power of two dividing each element of the Laplacian eigenvalue support of $u$. If $2^{e+1}$ divides $(n-1)(1+k)$, then $G\tilde{\circ}\overrightarrow{H}$ has LPGST between vertices $(u,0)$ and $(v,0)$.
\end{theorem}

\begin{proof}
Let $S=\mathrm{{supp}}_{L_G}(u)$ be the Laplacian eigenvalue support of $u$ in $G$. Then $F_\theta(G)\mathbf{e}_u^{n}=\mathbf{0}$ for all Laplacian eigenvalues $\theta \notin S$ of $G$, which implies that $(\mathbf{e}^{n}_u)^\top F_\theta(G)\mathbf{e}^{n}_v=\mathbf{0}$ for all $\theta \notin S$. By Proposition \ref{the_spectral}, we have
 \begin{align*}
 &\mathbf{e}^\top_{(u,0)}\exp(-\mathrm{i}t L_{G\tilde{\circ}\overrightarrow{H} })\mathbf{e}_{(v,0)}\nonumber\\[0.2cm]
 =&e^{-\mathrm{i}t\frac{(n-1)(1+k)}{2}}\left(\sum_{\theta \in S\setminus\left\{0\right\}}
e^{-\mathrm{i}t\frac{\theta}{2}}(\mathbf{e}^{n}_u)^\top F_\theta(G)\mathbf{e}^{n}_v
\left(\cos\left(\frac{\Delta_\theta t}{2}\right)+
\mathrm{i}\frac{(n-1)(1-k)-\theta}{\Delta_\theta}
\sin\left(\frac{\Delta_\theta t}{2}\right)\right)\right)\\[0.2cm]
&+\frac{k}{n(k+1)}e^{-\mathrm{i}t(n-1)(1+k)}+\frac{1}{n(k+1)}.
\end{align*}
In order to prove that $G\tilde{\circ}\overrightarrow{H}$ has LPGST from $(u,0)$ to $(v,0)$, we first prove that there exists a time $T$ such that
\begin{equation}\label{spec1}
	e^{-\mathrm{i}\frac{(n-1)(1+k)T}{2}}=1,
	\end{equation}
\begin{equation}\label{spec2}
	\left|\sum_{\theta \in S}
e^{-\mathrm{i}\frac{\theta T}{2}}(\mathbf{e}^{n}_u)^\top F_\theta(G)\mathbf{e}^{n}_v\right|=1,
	\end{equation}
and
\begin{equation}\label{spec3}
	\cos\left(\frac{\Delta_\theta T}{2}\right)\approx 1,
	\end{equation}
where $\Delta_\theta=\sqrt{((n-1)(k-1)+\theta)^2+4k}$, for each Laplacian eigenvalue $\theta\in S$ of $G$.

Note that, by Theorem \ref{Coutinho}, all eigenvalues in $S$ are integers. Consider a positive integer $\theta\in S\setminus\left\{0\right\}$. By the proof of Theorem \ref{nonbip1}, we obtain that $\Delta_\theta$ is irrational. Set $\Delta_\theta:=a_\theta\sqrt{b_\theta}$, where $a_\theta, b_\theta \in \mathbb{N}$ and  $b_\theta$ is the square-free part of $\Delta^2_\theta$. By Corollary \ref{Ri1}, the disjoint union $$\{1\}\cup \left\{\sqrt{b_\theta}:~\theta\in S, ~ \theta>0\right\}$$
     is linearly independent over $\mathbb{Q}$. By Theorem \ref{H-W}, there exist integers $d_\theta$ for each $\theta \in S$, and an integer $\kappa$ such that
     \begin{equation}\label{spec4}
	\kappa\sqrt{b_\theta}-d_\theta\approx -\frac{\sqrt{b_\theta}}{2^{e+1}}.
	\end{equation}
If $b_\theta=b_\nu$ for two distinct eigenvalues $\theta$ and $\nu$ in the support of $u$, then $d_\theta=d_\nu$. Multiplying both sides of (\ref{spec4}) by $4a_\theta$ yields
$$(4\kappa+2^{1-e})\Delta_\theta\approx 4a_\theta d_\theta. $$
Thus, if $T=(4\kappa+2^{1-e})\pi$, then
$$
\cos\left(\frac{\Delta_\theta T}{2}\right)\approx 1.
$$
So, we can find $T=(4\kappa+2^{1-e})\pi$ such that $(\ref{spec3})$ is valid.

Recall that $G$ has LPST between vertices $u$ and $v$, and recall also that $2^e$ is the greatest power of two dividing each element of $S$. By Theorem \ref{Coutinho}, $G$ has LPST at time $\frac{\pi}{2^e}$.
Then
$$
\left|\sum_{\theta \in S}
e^{-\mathrm{i}\frac{\theta T}{2}}
(\mathbf{e}^{n}_u)^\top F_\theta(G)\mathbf{e}^{n}_v\right|=\left|\sum_{\theta \in \mathrm{Spec}_L(G)}
e^{-\mathrm{i}\theta\frac{\pi}{2^e}}
(\mathbf{e}^{n}_u)^\top F_\theta(G)\mathbf{e}^{n}_v\right|=1,
$$
which implies that (\ref{spec2}) is true if $T=(4\kappa+2^{1-e})\pi$.

Note that  $2^{e+1}$ divides $(n-1)(1+k)$. Then
\begin{align*}
  \exp\left(-\mathrm{i}\frac{(n-1)(1+k)T}{2}\right) &=\exp\left(-\mathrm{i}\frac{(n-1)(1+k))(4\kappa+2^{1-e})\pi}{2}\right) \\
  &=\exp\left(-\mathrm{i}\frac{(n-1)(1+k)(4\kappa2^e+2)\pi}{2^{e+1}}\right)\\
  &=1,
\end{align*}
that is, (\ref{spec1})  is right if $T=(4\kappa+2^{1-e})\pi$.

Thus, (\ref{spec1}), (\ref{spec2}) and (\ref{spec3}) are valid if $T=(4\kappa+2^{1-e})\pi$, and then we have
 \begin{align*}
 &\left|\mathbf{e}^\top_{(u,0)}\exp(-\mathrm{i}T L_{G\widetilde{\circ}\overrightarrow{H} })\mathbf{e}_{(v,0)}\right|\\
 \approx & \left|\left(\sum_{\theta \in S\setminus\{0\}}
e^{-\mathrm{i}\frac{\theta}{2}T}(\mathbf{e}^{n}_u)^\top F_\theta(G)\mathbf{e}^{n}_v
\right)+\frac{k}{n(k+1)}e^{-\mathrm{i}T(n-1)(1+k)}+\frac{1}{n(k+1)}\right|\\
=& \left|\left(\sum_{\theta \in S\setminus\{0\}}
e^{-\mathrm{i}\frac{\theta}{2}T}(\mathbf{e}^{n}_u)^\top F_\theta(G)\mathbf{e}^{n}_v\right)
+\frac{k}{n(k+1)}+\frac{1}{n(k+1)}\right|\\
=&  \left|\sum_{\theta \in S\setminus\{0\} }
e^{-\mathrm{i}\frac{\theta}{2}T}(\mathbf{e}^{n}_u)^\top F_\theta(G)\mathbf{e}^{n}_v+\frac{1}{n}\right|\\
=&  \left|\sum_{\theta \in S\setminus\{0\} }
e^{-\mathrm{i}\frac{\theta}{2}T}(\mathbf{e}^{n}_u)^\top F_\theta(G)\mathbf{e}^{n}_v+e^{-\mathrm{i}\frac{0}{2}T}(\mathbf{e}^{n}_u)^\top F_0(G)\mathbf{e}^{n}_v\right|\\
= &\left|\sum_{\theta \in S}
e^{-\mathrm{i}\theta\frac{\pi}{2^e}}
(\mathbf{e}^{n}_u)^\top F_\theta(G)\mathbf{e}^{n}_v\right|\\
=& 1.
\end{align*}
Therefore, if $T=(4\kappa+2^{1-e})\pi$, $G\tilde{\circ}\overrightarrow{H}$ has LPGST between vertices $(u,0)$ and $(v,0)$.
\qed
\end{proof}

\begin{example}
\label{pregol}
{\em Let $K_2$ denote the path on $2$ vertices. Label the vertices of $K_2$ by $u$ and $v$. Note \cite{Cohl} that $K_2$ has LPST between $u$ and $v$. It is easy to verify that the Laplacian eigenvalues of $K_{2}$ are $0$ and $2$, which are both in the Laplacian eigenvalue support of $u$. Thus, $2^e=2$. Let $\overrightarrow{H}=(H_1, H_2)$ an $2$-tuple of graphs with $|V(H_i)|=k\geq 1$, $i=1,2$. If  $2^{e+1}\mid (k+1)$, that is, $k \equiv 3~(\mod~4)$, then by Theorem \ref{pregoo-2}, $K_{2}\tilde{\circ}\overrightarrow{H}$ has LPGST  between $(u,0)$ and $(v,0)$.

In particular, if $H_1$ and $H_2$ are two edgeless graphs with $k$ vertices, where $k \equiv 3~(\mod~4)$, then $K_{2}\tilde{\circ}\overrightarrow{H}$ is a tree, which has LPGST. Here, we give an example on trees having LPGST, keeping in mind \cite{Cohl} that every tree with at least three vertices has no LPST.

}
\end{example}

%
%\bigskip
%
%\noindent \textbf{Acknowledgements}~~The authors greatly appreciate the anonymous referees for their comments and suggestions.


\begin{thebibliography}{99}

{\small

\bibitem{Ack}
E. Ackelsberg, Z. Brehm, A. Chan, J. Mundinger,  C. Tamon, Laplacian state transfer in coronas, Linear Algebra Appl. 506 (2016)  154--167.

\bibitem{AckBCMT16}
E. Ackelsberg, Z. Brehm, A. Chan, J. Mundinger,  C. Tamon, Quantum state transfer in coronas, Electron. J. Combin. 24 (2) (2017) \#P2.24.

\bibitem{AlvirDLM16}R. Alvir, S. Dever, B. Lovitz, J. Myer, C. Tamon, Y. Xu, H. Zhan, Perfect state transfer in Laplacian quantum walk, J. Algebraic Combin. 43 (4) (2016) 801--826.



\bibitem{Angeles10}
R. J. Angeles-Canul, R. M. Norton, M. C. Opperman, C. C. Paribello, M. C. Russell, C. Tamon, Perfect state transfer, integral circulants, and join of graphs, Quantum Inf. Comput. 10 (3\&4) (2010)  0325--0342.

\bibitem{Ban}
 L. Banchi, G. Coutinho, C. Godsil,  S. Severini, Pretty good state transfer in qubit chains--The Heisenberg Hamiltonian, J. Math. Phys. 58 (2017) 032202.





\bibitem{Basic11}
M. Ba{\v{s}}i{\'c}, Characterization of circulant networks having perfect state transfer, Quantum Inf. Process. 12 (1) (2013) 345--364.

\bibitem{Bose03} S. Bose, Quantum communication through an unmodulated spin chain, Phys. Rev. Lett. 91 (20) (2003) 207901.

\bibitem{SBose}
S. Bose, A. Casaccino, S. Mancini,  S. Severini, Communication in XYZ all-to-all quantum networks with a missing link, Int. J. Quantum Inf. 7 (4) (2009) 713--723.


\bibitem{CaoCL20} X. Cao, B. Chen,  S. Ling, Perfect state transfer on Cayley graphs over dihedral groups: the non-normal case, Electron. J. Combin. 27 (2) (2020) P2.28.

\bibitem{CaoF21} X. Cao, K. Feng, Perfect state transfer on Cayley graphs over dihedral groups, Linear Multilinear Algebra 69 (2) (2021) 343--360.

\bibitem{CaoWF20}  X. Cao, D. Wang, K. Feng, Pretty good state transfer on Cayley graphs over dihedral groups, Discrete Math. 343 (1)  (2020)  111636.

\bibitem{CC}
W.-C. Cheung, C. Godsil, Perfect state transfer in cubelike graphs, Linear Algebra Appl. 435 (2011) 2468--2474.


\bibitem{chris2}
M. Christandl, N. Datta, T. Dorlas, A. Ekert, A. Kay, A. Landahl, Perfect transfer of arbitrary states in quantum spin networks, Phys. Rev. A 71 (2005) 032312.


\bibitem{Coutinho16}
G. Coutinho, Quantum State Transfer in Graphs, PhD thesis, University of Waterloo, 2014.

\bibitem{Coutinho15}
G. Coutinho, C. Godsil, K. Guo, F. Vanhove, Perfect state transfer on distance-regular graphs and association schemes, Linear Algebra Appl. 478 (2015) 108--130.

\bibitem{Coh12}
G. Coutinho, K. Guo, C. M. van Bommel, Pretty good state transfer between internal nodes of paths, Quantum Inf. Comput. 17 (9\&10) (2017) 0825--0830.

\bibitem{Cohl}
G. Coutinho, H. Liu, No Laplacian perfect state transfer in trees, SIAM J. Discrete Math. 29 (4) (2015) 2179--2188.



\bibitem{kn:Cui12}
 S.-Y. Cui, G.-X. Tian, The spectrum and the signless Laplacian spectrum of coronae, Linear Algebra Appl. 437 (2012) 1692--1703.


\bibitem{Fan}
X. Fan, C. Godsil, Pretty good state transfer on double stars, Linear Algebra Appl. 438 (5) (2013) 2346--2358.


\bibitem{FarhiG98}
E. Farhi, S. Gutmann,  Quantum computation and decision trees, Phys. Rev. A
    58 (3) (1998) 915--928.


\bibitem{Godsil11}
C. Godsil, Periodic graphs, Electron. J. Combin. 18 (1) (2011) \#P23.

\bibitem{CGodsil}
C. Godsil, State transfer on graphs, Discrete Math. 312 (1) (2012) 129--147.


\bibitem{Godsil12}
C. Godsil, When can perfect state transfer occur?,  Electron. J. Linear Algebra 23 (2012) 877--890.

\bibitem{Godsil1}
C. Godsil, S. Kirkland, S. Severini, J. Smith, Number-theoretic nature of communication in quantum spin systems, Phys. Rev. Lett. 109 (5) (2012) 050502.


\bibitem{Hall13}
B. C. Hall, Quantum Theory for Mathematicians, Springer, New York, 2013.


\bibitem{Hw}
G. H. Hardy, E. M. Wright, An Introduction to the Theory of Numbers, Fifth Edition,
Oxford University Press, 2000.




\bibitem{KirklandS11} S. Kirkland, S. Severini, Spin-system dynamics and fault detection in threshold networks, Physical Review A 83 (1) (2011) 012310.



\bibitem{LiLZ20-1}
Y. Li, X. Liu, S. Zhang, Laplacian state transfer in edge coronas, Linear Multilinear Algebra, available online at \url{https://doi.org/10.1080/03081087.2020.1751034}, 24 pages, 2020.


\bibitem{LiLZ20-2}
Y. Li, X. Liu, S. Zhang, Laplacian state transfer in $Q$-graph, Appl. Math. Comput. 384 (2020) 125370.

\bibitem{LiLZ2021}
Y. Li, X. Liu, S. Zhang, Laplacian perfect state transfer in extended neighborhood coronas, Acta Math. Sin., English Series 37(12) (2021) 1921--1932.

\bibitem{LiLZZ2021}
Y. Li, X. Liu, S. Zhang,  S. Zhou, Perfect state transfer in NEPS of complete graphs, Discrete Appl. Math. 289 (2021) 98--114.

\bibitem{LiuW2021} X. Liu, Q. Wang, Laplacian state transfer in total graphs, Discrete Math. 344(1) (2021) 112139.


\bibitem{LiuZ19}
X. Liu, Z. Zhang, Spectra of subdivision-vertex join and subdivision-edge join of two graphs, B. Malays. Math. Sci. So.  42 (2019) 15--31.


\bibitem{kn:McLeman11}
C. McLeman, E. McNicholas, Spectra of coronae, Linear Algebra Appl. 435 (2011) 998--1007.


\bibitem{MAALA}
C. D. Meyer, Matrix Analysis and Applied Linear Algebra, Siam, Philadelphia, 2000.



%\bibitem{HPal3}
%H. Pal, More circulant graphs exhibiting pretty good state transfer, Discrete Math.  341 (4) (2018) 889--895.

%\bibitem{HPal}
%H. Pal,  B. Bhattacharjya, Perfect state transfer on gcd-graphs, Linear Multilinear Algebra 65 (11) (2017) 2245--2256.

%\bibitem{HPal1}
%H. Pal,  B. Bhattacharjya, Perfect state transfer on NEPS of the path on three vertices, Discrete Math. 339 (2016) 831--838.

%\bibitem{HPal5}
%H. Pal,  B. Bhattacharjya, Pretty good state transfer on circulant graphs, Electron. J. Combin. 24 (2) (2017) \#P2.23.

%\bibitem{HPal4}
%H. Pal,  B. Bhattacharjya, Pretty good state transfer on some NEPS, Discrete Math. 340 (4) (2017) 746--752.





\bibitem{Ri}
I. Richards, An application of Galois theory to elementary arithmetic,  Adv. Math.  13 (3) (1974) 268--273.

\bibitem{TanFC19} Y. Tan, K. Feng, X. Cao, Perfect state transfer on abelian Cayley graphs, Linear Algebra Appl. 563 (2019) 331--352.


\bibitem{WangC} D. Wang, X. Cao, Pretty good state transfer on Cayley graphs over semi-dihedral groups, Linear Multilinear Algebra, available online at \url{https://doi.org/10.1080/03081087.2021.1926414}, 16 pages, 2021.


\bibitem{WangL2021} J. Wang, X. Liu, Laplacian state transfer in edge complemented coronas, Discrete Appl. Math. 293 (2021) 1--14.

\bibitem{SZ}
S. Zheng, X. Liu, S. Zhang, Perfect state transfer in NEPS of some graphs, Linear
Multilinear Algebra 68(8) (2020) 1518--1533.



\bibitem{Zhang05}
F. Zhang, The Schur Complement and Its Applications, Springer, New York, 2005.


\bibitem{HZ}
J. Zhou, C. Bu, J. Shen, Some results for the periodicity and perfect state transfer, Electron. J. Combin. 18 (2011) \#P184.


\bibitem{Zhou14}
J. Zhou, C. Bu, State transfer and star complements in graphs, Discrete Appl. Math. 176 (2014) 130--134.


}

\end{thebibliography}
\end{document}